\documentclass[reqno]{amsart}
\usepackage{amssymb}
\usepackage{amsmath}
\usepackage{amsthm}
\usepackage{pgf}
\usepackage{color}
\usepackage{graphicx}   
\usepackage{hyperref}

\newtheorem{theorem}{Theorem}[section]

\newtheorem{proposition}[theorem]{Proposition}

\theoremstyle{definition}

\numberwithin{equation}{section}

\newcommand{\Rz}{{\mathbb R}}
\newcommand{\Nz}{\mathbb{N}}
\newcommand{\Zz}{\mathbb{Z}}

\newcommand{\haz}{\widehat}
\newcommand{\ove}{\overline}

\begin{document}

\title[Graphene ground states]{Graphene ground states}

\author{Manuel Friedrich}
\address[Manuel Friedrich]{Faculty of Mathematics, University of Vienna, 
Oskar-Morgenstern-Platz 1, 1090 Wien, Austria.}
\email{manuel.friedrich@univie.ac.at}
\urladdr{http://www.mat.univie.ac.at/$\sim$friedrich}

\author{Ulisse Stefanelli}
\address[Ulisse Stefanelli]{Faculty of Mathematics, University of Vienna, 
Oskar-Morgenstern-Platz 1, 1090 Wien, Austria  and  Istituto di Matematica
Applicata e Tecnologie Informatiche \textit{{E. Magenes}}, v. Ferrata 1, 27100
Pavia, Italy.}
\email{ulisse.stefanelli@univie.ac.at}
\urladdr{http://www.mat.univie.ac.at/$\sim$stefanelli}


\begin{abstract} Graphene is locally two-dimensional but not flat. Nanoscale ripples appear in
suspended samples and rolling-up often occurs when boundaries
are not fixed.  We address this variety of graphene geometries by classifying all ground-state deformations of the
  hexagonal lattice  with respect to configurational energies including
  two- and three-body  terms.  As a consequence, we   prove that all ground-state deformations are either  periodic
in one direction, as in the case of ripples, or rolled up, as in the case of nanotubes.
\end{abstract}

\maketitle

\section{Introduction}

Graphene is a one-atom thick layer of carbon atoms arranged in  a  regular
hexagonal lattice. Its serendipitous discovery in 2005 sparkled
research on two-dimensional materials systems. This new branch of
Materials Science exponentially developed in the last years. An
impressive variety of new low-dimensional systems  has  been presented
and their potential for innovative applications, especially in
optoelectronics, is currently strongly investigated \cite{Ferrari}.

The lower-dimensionality of graphene is at the basis of its amazing mechanical, optical, and electronic properties. On the other hand, the classical Mermin-Wagner Theorem \cite{Landau2,Mermin,Mermin2} excludes the possibility of realizing truly two-dimensional systems at finite temperature. 
 Indeed, observations on suspended samples seem to indicate that
 graphene is generally not exactly flat but gently rippled
 \cite{Meyer}. Wavy patterns on the scale of approximately one hundred
 atom spacings have been  computationally investigated \cite{Fasolino}
 and   are considered to be  responsible for the stabilization of graphene at
 finite temperature.  Nonplanarity is expected even in the zero-temperature limit, due to quantum
 fluctuations \cite{Herrero}. 
The Reader is referred to the recent  survey 
 \cite{Deng} for an overview of  ripple-formation  mechanisms and possible
 applications.
On the other hand, free graphene samples in
 absence of support have the tendency to roll-up in tube-like
 structures \cite{Lambin}. 

The phenomenon of  rippling and rolling-up in graphene is here tackled
from the molecular-mechanical viewpoint. The actual configuration of a
graphene sheet is identified with a three-dimensional deformation of
the ideal hexagonal lattice. To each deformation we associate a
configurational energy   which takes nearest-neighbor and  next-to-nearest-neighbor
two-body interactions \cite{Brenner90, Stillinger,Tersoff}  into account and favors locally  the specific bonding mode in graphene.

Our main result is a complete classification of ground-state deformations. We
show that  such  ground states are locally not flat, as specific nonplanar
optimal configurations ensue. In particular, two different optimal
configurations for single hexagonal cells are identified. Geometric
compatibility forces these optimal cells to combine in specific
patterns in order to give rise to global deformations. This fact
allows us to  classify  ground states, which correspond either to
rippled or to rolled-up structures, see Theorem \ref{th:characterization}. 

Before closing this introduction, let us review the literature on the
mathematical modeling of graphene via Molecular Mechanics.
The first {\it global-minimality} result for graphene in two
dimensions has to be traced
back to {\sc E \& Li} \cite{E-Li09} who investigate the so-called {\it thermodynamic
  limit} as the number of atoms  tends  to infinity. Their result corresponds
to  an  extension of the seminal theory by {\sc Theil} \cite{Theil06} to
three-body interaction energies favoring $2\pi/3$ bond angles.  More recently, {\sc Farmer,
  Esedo\={g}lu, \& Smereka} \cite{Smereka15} obtained an analogous
 result  by assuming the
three-body energy term to favor $\pi$ bond angles, which calls
for the minimality of graphene among frustrated configurations. 

In case of a {\it
  finite} number of atoms in two dimensions, graphene patches are
identified as the only ground states in
\cite{Mainini-Stefanelli12} and  are  characterized in terms  of  a discrete isoperimetric
inequality in \cite{Davoli15}. The emergence of a
hexagonal Wulff shape as the number of atoms
increases can be also quantitatively checked  \cite{Davoli15}. 

If one allows the
configuration to be three-dimensional, flat graphene is no more
expected to be a ground state \cite{Mainini-Stefanelli12}. By reducing
to nearest-neighbor interactions, it can
nonetheless be checked to be a local minimizer, under specific
assumptions on the interaction potentials \cite{stable}. This stability
analysis allows to tackle other carbon nanostructures as well,
including nanotubes \cite{tube,numeric-stability,stability},
fullerenes \cite{Friedrich16,stable}, diamond \cite{stable}, 
carbyne  stratified configurations  \cite{Lazzaroni17}.

As  concerns  rippling, one has to mention   the recent  paper 
\cite{Davini} where the Gaussian stiffness of graphene, namely its
tendency to favor   non-null    Gaussian-curved configurations, is investigated via a
discrete-to-continuum procedure. The  aim there is  to obtain an
analytical expression for  the Gaussian stiffness by focusing on a
specific choice of the functional. 
 In contrast, our focus is here on energetics and global geometries of
 ground states under general qualitative 
 assumptions on the configurational energy.

The occurrence of nonflat and rolled-up ground states  can be
avoided by additionally imposing periodic boundary
conditions. Experimentally, this corresponds to clamp the edges
of a suspended graphene sample. 
In this case, by
extending the energy to include  third-neighbor  interactions, we prove
in the companion paper
\cite{Friedrich17} that some specific optimal ripple length can be
identified, independently of the  sample size.   This provides an analytical validation to
the computational findings in \cite{Fasolino}.

\section{Energy}\label{sec:energy}

The focus of this paper is on global minimization in three
dimensions. We restrict the class of admissible configurations to
 deformations $y :  H \to \Rz^3$ of the {\it hexagonal lattice}
\begin{align*}
H&=\{sa+tb+rc\;:\; s,\, t\in \Zz, \, r=0,1\}
\end{align*}
where $a=(3/2,\sqrt{3}/2)$, $b=(0,\sqrt{3})$, and $c=(1,0)$.  In particular, the reference
configuration  as well as all  atom  coordinations (neighbors) are  kept fixed. We call $a$, $b$, and $a-b$ {\it coordinate directions of} $H$ and term  {\it
  hexagonal graph} the graph connecting all first neighbors in $H$. A {\it reference cell} is any  $\{x_1,\dots,x_6\}$ corresponding to a simple cycle in the hexagonal
graph and we call {\it cell} its image $\{y_1,\dots,y_6\}$ through
$y$, namely $y_i=y(x_i)$. The labeling of the atoms in each reference
cell is always meant to be arranged counterclockwise with $x_1=na+mb$
to be such that $n+m$ is minimal in the reference cell, see Figure
\ref{deformation}. 
\begin{figure}[h]
  \centering
  \pgfdeclareimage[width=140mm]{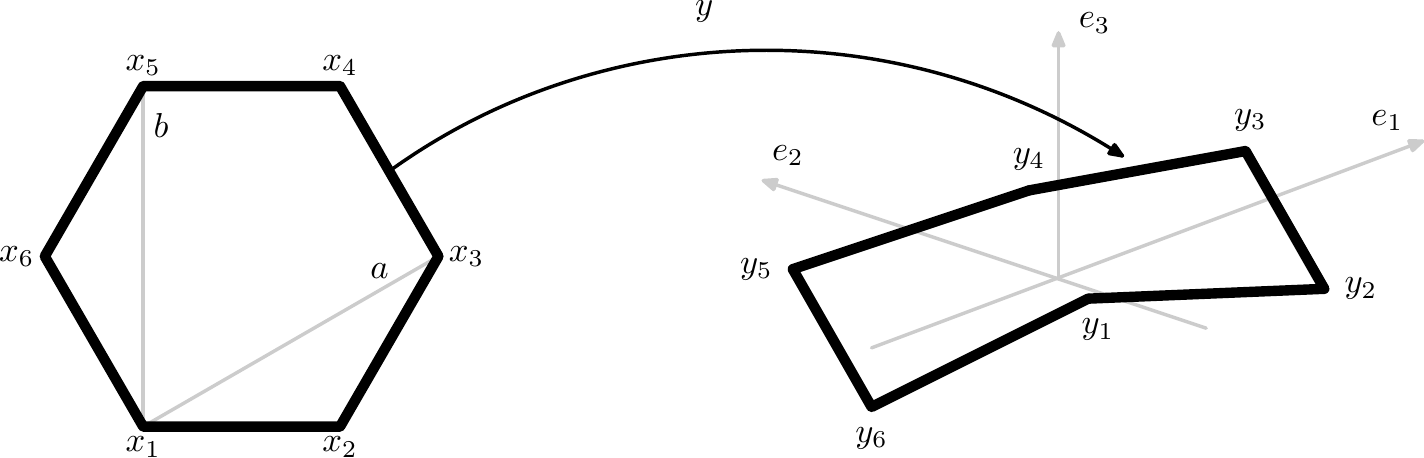}{deformation} 
\pgfuseimage{deformation}
\caption{The reference cell $\{x_1, \dots,x_6\}$ (left) and the
  cell $\{y_1,\dots,y_6\}$ (right).}\label{deformation}
\end{figure}

The {\it cell energy} of the cell $\{y_1,\dots,y_6\}$ is given by 
$$E_{\rm cell}(y_1,\dots,y_6) = \frac12 \sum_{i=1}^6 v_2(|y_i -y_{i-1}|) + \sum_{i=1}^6 v_2(|y_i -y_{i-2}|)+  \sum_{i=1}^6v_3(\theta_i),$$
where sums in the indices are meant modulo $6$ throughout. The first term corresponds to   {\it
  nearest-neighbors} and the second term to  {\it next-to-nearest-neighbors}.  The factor $1/2$
reflects the fact that each segment $\{y_{i+1},y_i\}$, called {\it bond} in the following, is contained in two adjacent hexagonal cells. With  $\theta_i$ we indicate  the {\it bond angle}
at $y_i$ formed by the segments $\{y_{i+1},y_i\}$ and $\{y_{i-1},y_i\}$ which is less or equal to $ \pi$.

We assume that the   {\it two-body}  interaction potential
$v_2:\Rz^+\to[-1,\infty)$ attains its minimum value  only   at
$1$ with $v_2(1) = -1$. Moreover, we suppose that
$v_2$ is continuous and decreasing on $(0,1)$ (i.e., \emph{short-range repulsive}) and $v_2$
increasing on $[1,\infty)$ (\emph{long-range attractive}).   Furthermore, we suppose  that $v_2$   is differentiable in   $ (5/4,\sqrt{3}]$   with
$v_2' >0$.  
The  {\it three-body} interaction density
$v_3: [0,\pi]\to[0,\infty)$
is assumed to be continuous and to attain the minimum value $0$
only at $2\pi/3$  where it is differentiable.  These basic assumptions correspond to the fact
that $sp2$ covalent bonds in carbon 
are characterized by some reference bond length, here normalized to 1,
and a reference bond angle of amplitude $2\pi/3$ \cite{Clayden12}.
Note that $E_{\rm cell}$  has a bounded sublevel  (among cells with barycenter zero) 
 and is continuous. As such, it
admits minimizers, which we call {\it optimal cells}. These will be
characterized in Proposition \ref{CZ}.   For a fine
characterization of the minimizers,   some additional qualification on
$v_2$ and $v_3$ will be needed,  see conditions   \eqref{minimal}-\eqref{convex} below.

We identify the deformation $y:H \to \Rz^3$ with the collection of its cells. Furthermore, cells are identified via the inverse of  $y$ to their reference cells and these are labeled in terms of their barycenters. Indeed, barycenters of reference cells form the {\it triangular lattice}
\begin{align*}
T&=(1/2,\sqrt{3}/2)+\{sa+tb\;:\; s,\, t\in \Zz\}.
\end{align*}
We will hence equivalently indicate cells as $\{y_1, \dots,y_6\} \in
(\Rz^3)^6$ or $(s,t)\in \Zz^2$, where $sa+tb$ is the barycenter of the corresponding reference cell $\{x_1,\dots,x_6\}$.

The {\it energy} of the deformation $y:H \to \Rz^3$ is then defined as
$$E(y) = \sup_{m\in \Nz}\left(\frac{1}{\# (T\cap B_m)} \sum_{(s,t)\in  T\cap B_m} E_{\rm cell} (s,t)\right)$$
where $B_m \subset \Rz^2$ is the ball centered at $0$ having radius
$m$. A deformation is called a {\it ground state} if it minimizes  the energy $E$. Note that $E$ corresponds to the supremum of cell-energy densities on bounded sets of cells. This immediately entails the following. 

\begin{proposition}[Only optimal cells]
A deformation is a ground state if and only if all its cells are optimal.
\end{proposition}

\begin{proof}[\bf Proof]
By letting $E^* = \min E_{\rm cell}$, we readily check that $E \geq
E^*$. If all cells are optimal, we have $E=E^*$ and the deformation is
a ground state. On the other hand, let $E=E^*$ and assume by contradiction that the cell $(s,t)\in T\cap B_m$ is not optimal. Then,
$$E \geq \frac{1}{\# (T\cap B_m)} \sum_{(s,t)\in  T\cap B_m} E_{\rm cell} (s,t) > E^*,$$ 
contradicting minimality.
\end{proof}

In the following, some quantitative specifications on the interaction
densities $v_2$ and $v_3$ will be assumed. These are intended to
ensure that optimal cells indeed have a hexagonal-like shape.
In particular, we will ask for a small parameter   $0<\delta\leq
0.2$   such that 
\begin{align}
&v_2(1-\delta)  > 11 + 12  v_2(\sqrt{3}),\label{minimal}\\
&v_2(1+\delta)  >  - 1 +  12  v_2 (\sqrt{3}) -  12  v_2(\sqrt{3}(1-\delta)^2), \label{maximal}\\
& v_3(\theta)  >   6 + 6 v_2
(\sqrt{3}) \quad \text{if} \ \ | \theta -2\pi/3 |\geq \delta, \label{angle}\\
& (\ell_1, \ell_2,\theta) \mapsto  \frac{1}{4}  v_2(\ell_1) +  \frac{1}{4} 
v_2(\ell_2) + v_2 \big((\ell_1^2 + \ell_2^2 - 2 \ell_1\ell_2 \cos
\theta)^{1/2}\big) + v_3 (\theta) \nonumber\\
& \ \ \ \ \ \ \ \  \ \ \ \ \ \ \ \ \text{is strictly convex for} \ \
|\ell_1 -1|<\delta, \ |\ell_2 -1|<\delta, \ |\theta-2\pi/3|<\delta. \label{convex}
\end{align}
Properties \eqref{minimal}-\eqref{maximal} entail that
first-neighbor bond lengths range between $1-\delta$ and $1+\delta$
(note that $\sqrt{3}$ is the second-neighbor distance in $H$),
whereas \eqref{angle} ensures that the bond angles of the
optimal cell are $\delta$-close to $2\pi/3$. Eventually, assumption
\eqref{convex} yields that the contribution of first-neighbors is
strong enough to entail the symmetry of the optimal cell,
see Proposition \ref{optimal}.

Assumptions
\eqref{minimal}-\eqref{convex} will be tacitly assumed in the rest of
the paper. Note that these are compatible with a choice of densities
$v_2$ and $v_3$ growing sufficiently fast out of their minima  and
$v_2$ is sufficiently flat but increasing around $\sqrt{3}$. In
particular, the quantitative assumptions  on  $v_2$  introduced by {\sc Theil}
\cite{Theil06} (see also \cite{E-Li09,Smereka15}) imply \eqref{minimal}-\eqref{maximal}.
As a
matter of illustration, one can choose  the Lennard-Jones-like
potential $v_2$ and the Tersoff term $v_3$ \cite{Tersoff}
$$v_2(\ell) =   (a-1)  \,\ell^{-a} - a\, \ell^{ -a+1 }  \qquad v_3(\theta) =
\kappa  (1/2 +\cos \theta)^2$$
with $\kappa$ large enough (note that $v_2$  has minimum
$-1$ in $\ell =1$). For instance, one can choose  $a=18$, 
$\kappa=600$, and $\delta=0.12$.

\section{Optimal cells}\label{sec:optimal}

The aim of this section is to prove that optimal cells have specific
bonds and angles. Such  a  property
will be used in Section \ref{CZ} in order to characterize completely 
optimal cells.

\begin{proposition}[Bonds and angles of optimal cells] \label{optimal}  All bonds of an optimal cell have length  $\ell^* \le 1$  and all angles have amplitude $\theta^*<2\pi/3$, where $\ell^*$ and $\theta^*$ are uniquely determined in terms of  the energy.
\end{proposition} 

\begin{proof}[\bf Proof]
Recall assumptions \eqref{minimal}-\eqref{convex} and  let
$\{y_1,\dots,y_6\}$ be an optimal cell. We first show that $|y_j -
y_{j-1}| \in ( 1- \delta, 1+ \delta)$ and $\theta_j \in (2\pi/3 - \delta, 2\pi/3 + \delta)$ for all $j=1,\dots, 6$. 

In case $|y_j -
y_{j-1}| \leq  1- \delta$ for some $j=1,\dots, 6$, one has that 
\begin{align*}
  E_{\rm cell} (y_1,\dots,y_6) &=  \frac12  \sum_{i=1}^6v_2(|y_i -
y_{i-1}|) +  \sum_{i=1}^6 v_2(|y_i - y_{i-2}|) + \sum_{i=1}^6
v_3(\theta_i)\\
&\stackrel{\rm (a)}{\geq}   \frac12  v_2(1-\delta) +   \frac12  \sum_{i\not = j}v_2(|y_i -
y_{i-1}|) +  \sum_{i=1}^6 v_2(|y_i - y_{i-2}|) + \sum_{i=1}^6
v_3(\theta_i)\\
& \stackrel{\rm (b)}{ \ge }    \frac12  v_2(1-\delta)  - \frac{5}{2} - 6  +  \sum_{i=1}^6
v_3(\theta_i)\\
& \stackrel{\eqref{minimal}}{ > }   -3  + 6 v_2(\sqrt{3}) =  E_{\rm cell} (x_1,\dots,x_6),
\end{align*}
where we have used that (a) $v_2$ is decreasing in $(0,1)$ and (b)
$v_2\geq -1$. This contradicts optimality as the reference cell
$\{x_1,\dots,x_6\}$, i.e. the identity deformation, would have strictly lower energy. We conclude that
all first-neighbor bonds have to have at least length $1-\delta$.

Assume now that some  bond angle $\theta_j$ is such that $|\theta_j -
2\pi/3|\geq \delta$. Then
\begin{align*}
  E_{\rm cell} (y_1,\dots,y_6) &=   \frac12  \sum_{i=1}^6v_2(|y_i -
y_{i-1}|) +  \sum_{i=1}^6v_2(|y_i - y_{i-2}|) + \sum_{i=1}^6
v_3(\theta_i)\\
&\geq  -  9 + v_3(\theta_j)  \stackrel{\eqref{angle}}{>}  -3  + 6 v_2(\sqrt{3}) =  E_{\rm cell} (x_1,\dots,x_6)
\end{align*}
which again contradicts optimality. We have hence proved that all bond
angles  $\theta$  
necessarily  satisfy  $| \theta  -
2\pi/3|<\delta$.

Basic trigonometry  together with the least size of the bond lengths and bond angles   ensures that second-neighbor bonds have at least
length 
\begin{align}
  &2(1-\delta) \sin(\pi/3 - \delta/2) = 2(1-\delta)
  \left(\frac{\sqrt{3}}{2} \cos(\delta/2) - \frac12
    \sin(\delta/2)\right)> \sqrt{3} (1-\delta)^2  > 1 \label{trig}
\end{align}
where we also used that   $0<\delta\leq 0.2$.  Assume now that  $|y_j -
y_{j-1}| > 1 + \delta$ for some $j=1,\dots, 6$. We have that 
\begin{align*}
  E_{\rm cell} (y_1,\dots,y_6) &= \frac12 \sum_{i=1}^6v_2(|y_i -
y_{i-1}|) +   \sum_{i=1}^6 v_2(|y_i - y_{i-2}|) + \sum_{i=1}^6
v_3(\theta_i)\\
&\stackrel{\rm (c)}{\ge}     \frac{1}{2}  v_2(1+\delta)   - \frac{5}{2}  + 6 v_2( \sqrt{3}(1-\delta)^2)  \stackrel{\eqref{maximal}}{>}    -3  + 6 v_2(\sqrt{3}) =  E_{\rm cell} (x_1,\dots,x_6),
\end{align*}
where we have used in (c) that all second-neighbor bonds have length
at least  $\sqrt{3} (1-\delta)^2$,  see \eqref{trig}, and $v_2 $ is
increasing in $(1,\infty)$. The latter inequality once again contradicts optimality and
we conclude that all first-neighbor bond lengths are at most
$1+\delta$.

We have proved that if $\{y_1,\dots,y_6\}$ is optimal,
first-neighbor bond lengths $\ell_i = |y_i - y_{i-1}|$ lie in 
$(1-\delta,1  +  \delta)$ and bond angles $\theta_i$  lie in  $(2\pi/3
- \delta, 2\pi/3 + \delta)$. We can now decompose the cell energy
$E_{\rm cell}$  and use the convexity assumption \eqref{convex} in
order to get that 
\begin{align}\label{inequality}
 E_{\rm cell} (y_1,\dots,y_6)  &= \sum_{i=1}^6 \left( \frac{1}{4}  v_2(\ell_i) +  \frac{1}{4} 
v_2(\ell_{i  +  1}) + v_2 ((\ell_i^2 + \ell_{i  +  1}^2 - 2 \ell_i\ell_{i  +  1} \cos
\theta_i)^{1/2}) + v_3 (\theta_i) \right) \notag \\
&\geq 6 \left(   \frac{1}{2} 
v_2(\ell^*) + v_2 ( \sqrt{2} \ell^*  (1 -\cos
\theta^*)^{1/2}) + v_3 (\theta^*) \right) 
\end{align}
where 
$$ \ell^*=\frac16( \ell_1 + \dots+ \ell_6), \quad  
\theta^*=\frac16( \theta_1 + \dots+ \theta_6).$$
As the  inequality in \eqref{inequality} is strict whenever $\ell_i\not = \ell^*$ or $\theta_i
\not = \theta^*$ for some $i=1,\dots,6$, all bonds of an optimal cell
have length $\ell^*$ and all angles have amplitude $\theta^*$. It
remains to check that  $\ell^* \le 1$  and $\theta^*<2\pi/3$. First, if we
had  $\ell^* > 1$,  one could reduce the energy in \eqref{inequality}
by reducing $\ell^*$ noting that $v_2$ is increasing in $(1,\infty)$
and recalling \eqref{trig}. This, however, would again contradict  
optimality.  On the other hand, we have that 
$$6\theta^* = \theta_1 + \dots+ \theta_6\leq 4\pi$$
as $4\pi$ is the sum of the internal angles of a planar hexagon. In
particular, the equality holds iff $\{y_1,\dots,y_6\}$ is
planar. Hence, we have that $\theta^*\leq 2\pi/3$. However, we can exclude that
$\theta^*=2\pi/3$ for in this case all second neighbors would have
distance  $\sqrt{3}\ell^* \in \sqrt{3}(1-\delta,1]\subset
\sqrt{3}(0.8,1]  \subset (5/4,\sqrt{3}]$.   As $v_2'(\sqrt{3}\ell^*)>0$ and
$v_3'(2\pi/3)=0$, one would then strictly  lower the energy in \eqref{inequality} by reducing $\theta^*$. 
\end{proof}

Before closing this section let us comment on the importance of the
condition $v_2'>0$ in a left neighborhood of $\sqrt{3}$. This has been
used in the proof of Proposition \ref{optimal} in order to check that
$\theta^*$ is strictly smaller than $2\pi/3$. Indeed, if $v_2'$ were
flat in a neighborhood of $\sqrt{3}$, which would correspond to the
case of purely first-neighbor interactions, one would find  $\theta^*=
2\pi/3$, $\ell^*=1$  \cite{stable}, and the optimal cell would be
planar. Correspondingly, the only ground state would be the hexagonal
lattice $H$. 

\section{The $Z$ and the $C$ cells}\label{sec:CZ}

In the  previous  section we have proved that all cells of a ground state
have all bonds of length $\ell^*$ and all bond angles $\theta^*$. The aim of
this section is to check that such properties determine the cell (up
to isometries). More precisely, 
Proposition \ref{CZ} below states that exactly two geometries are possible:
the {\it $Z$ cell} and
the {\it $ C$ cell}.  This naming refers  to the cell
shape, see Figure \ref{CZfigure}, and  has  been inspired by
\cite{Davini}, where this nomenclature is however used for triplets of
adjacent bonds.
\begin{figure}[h]
  \centering
  \pgfdeclareimage[width=140mm]{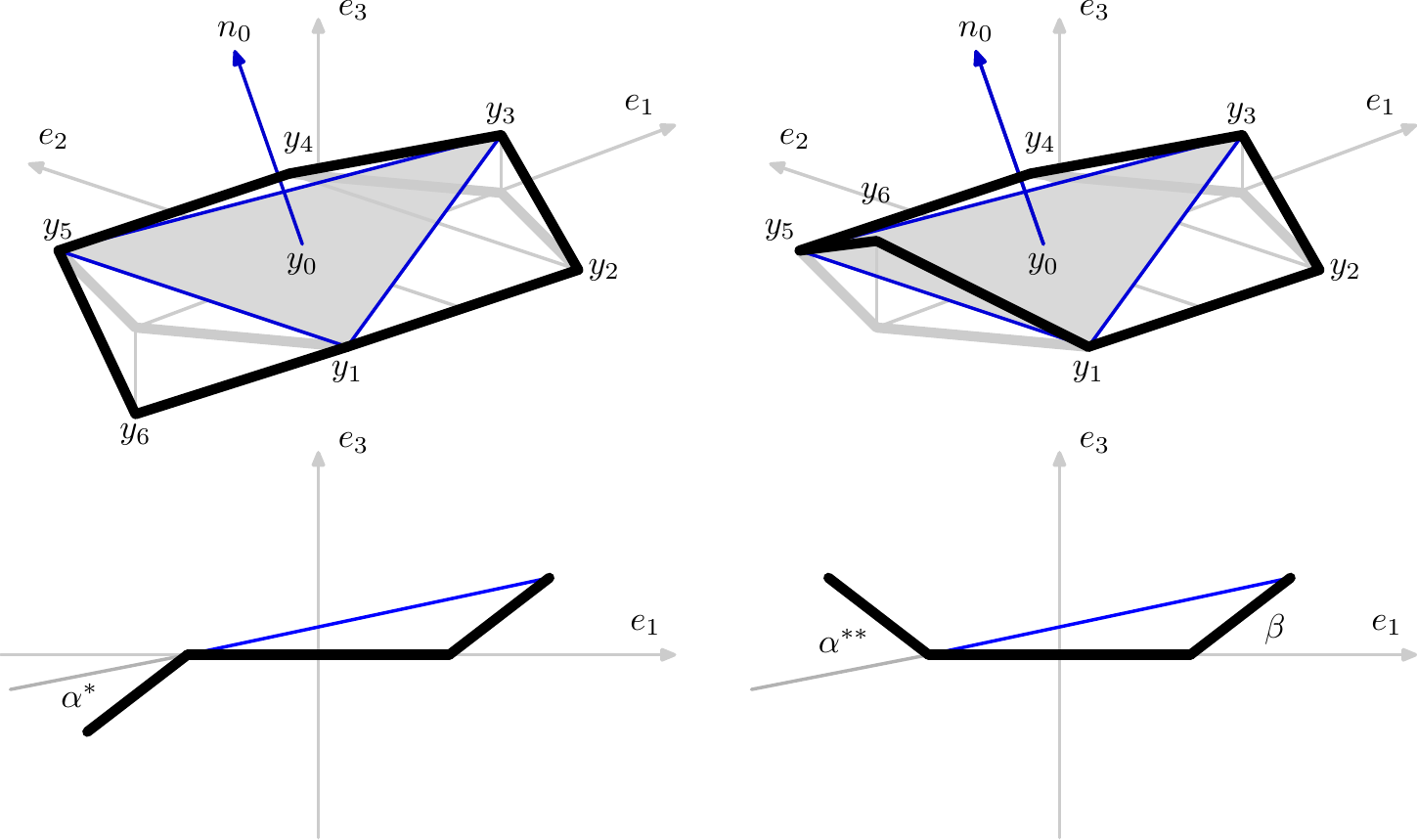}{CZcells} 
\pgfuseimage{CZcells}
\caption{The {$Z$ cell} (left) and {$C$ cell} (right), prospective views
  (top) and side views (bottom).  The normal vector $n_0$ to the darkened triangle $\lbrace y_1,y_3,y_5 \rbrace$,  its center $y_0$,  and the angles $\alpha_*$, $\alpha_{**}$ play an important role in the proof of Proposition \ref{CZ}. We will refer to the angle $\beta$ in Section \ref{sec:char}.  }\label{CZfigure}
\end{figure}
The $Z$ and the $C$ cell 
are specified as follows
\begin{align*}
  Z&=\{ (-\ell^*/2,-v,0),\, (\ell^*/2,-v,0), \, (\ell,0,h), \,
     (\ell^*/2,v,0)),\, (-\ell^*/2,v,0), \, (-\ell,0,-h)\},\nonumber\\
C&=\{ (-\ell^*/2,-v,0),\, (\ell^*/2,-v,0), \, (\ell,0,h), \,
     (\ell^*/2,v,0)),\, (-\ell^*/2,v,0), \, (-\ell,0,h)\}
\end{align*}
where $v$, $\ell$, and $h$ are given by 
\begin{equation}
v = \ell^*\left(\frac12 - \frac12 \cos \theta^*\right)^{1/2},\quad \ell = \frac{\ell^*}{2} -\ell^*\cos \theta^*, \quad h=\ell^*\left(\frac{1}{2} +\frac{1}{2}\cos\theta^*
-  \cos^2\theta^*\right)^{1/2}.  \label{h}
\end{equation}
These explicit values can be obtained by elementary (yet tedious)
trigonometry. Note that if $\theta^*$ were $2\pi/3$ (which is not), the
above formulas would give $v = (\sqrt{3}/2)\ell^*$, $\ell=\ell^*$,
and $h=0$, corresponding indeed to the flat hexagonal lattice of
spacing $\ell^*$. 


A remarkable property of the $Z$ and the $C$ cell is that they have a
pair of parallel bonds which  define  a plane with normal  $e_3$ containing four out of six atoms of the
cells. By considering the two semispaces divided by such plane, the
$Z$ and the $C$ cell are easily distinguishable as the two off-planar
atoms of $Z$ belong to two distinct semispaces, whereas those of $C$
belong to the same semispace. Both cells are symmetric with respect to
the  $(e_1,e_3)$  plane.  In addition, $Z$ is central symmetric as well as
invariant by $2\pi/3$ and $4\pi/3$ rotations about the axis with
direction $(y_3{-}y_1)\wedge(y_5{-}y_1)$  (i.e., direction of $n_0$ in Figure \ref{CZfigure}). 

The main result of this section is the following characterization.

\begin{proposition}[$C$ and $Z$ cells]\label{CZ}
  Optimal cells are either $C$ or $Z$.
\end{proposition}

\begin{proof}[\bf Proof]
  As bonds
  and bond angles of an optimal cell $  Y =  \{y_1,\dots,y_6\}$ are all equal, the distance of each pair of second neighbors
  is equal as well. In particular, the   three atoms $y_1$, $y_3$, and $y_5$ are the vertices of an
  equilateral triangle and  determine  a plane, which we indicate
  with $A$, see Figure \ref{CZfigure}. Fix an orientation
  on $A$ via the  unit vector  $n_0$ with direction
  $(y_3{-}y_1)\wedge(y_5{-}y_1)$ and    indicate with $\alpha_2$, $\alpha_4$, and
  $\alpha_6$ the incidence angles  with $A$ of the planes $A_2$, $A_4$, and
  $A_6$ containing $\{y_1,y_2,y_3\}$, $\{y_3,y_4,y_5\}$, and $\{y_5,y_6,y_1\}$,
  respectively.  More precisely,  let $n_2$, $n_4$, and
  $n_6$ be the  unit vectors  with directions $(y_3{-}y_2)\wedge (y_1{-}y_2)$, $(y_5{-}y_4)\wedge (y_3{-}y_4)$, and
  $(y_1{-}y_6)\wedge (y_5{-}y_6)$, respectively, and recall that $$\alpha_i
  = \arccos(n_0\cdot n_i) \in [0,\pi]\quad\text{for}\ \ i=2,4,6.$$ 
The geometry of the cell $Y$ is completely determined
  by the three incidence angles $\alpha_2$, $\alpha_4$, and
  $\alpha_6$ and  by  the sign of the products 
$$p_i=(y_i{-}y_0)\cdot n_0 \quad\text{for}\ \ i=2,4,6$$ 
where we have indicated by $y_0$ the  center of $A$, namely
$y_0=(y_1 +y_3 +y_5)/3$.  
 
In the case of the $Z$ cell, all $p_i$ have the
  same sign and all incidence angles  are all
  equal to 
  \begin{align}
   &  \alpha^* =   \arctan\left(  \frac{h}{\ell - \ell^*/2} \right) -  \arctan\left(\frac{h}{\ell+\ell^*/2}\right) 
    \nonumber
  \end{align}
which just depends on $\theta^*$, see \eqref{h}.
In particular, in the setting of Figure \ref{CZfigure} one has that  $p_i<0$ for
  $i=2,4,6$.

In case of the $C$ cell, one has that two out of three products $p_i$ have the same sign and the third has the opposite
  sign. The incidence angles $\alpha_i$ corresponding to the products
  $p_i$ with the same sign are $\alpha^*$ and that
  corresponding to the product with opposite sign equals
\begin{align}
    &\alpha^{**} = \arctan\left(  \frac{h}{\ell - \ell^*/2} \right) +  \arctan\left(\frac{h}{\ell+\ell^*/2}\right)  
    \nonumber
  \end{align}
which again depends on  $\theta^*$ only.
The setting of Figure \ref{CZfigure} corresponds to
 $p_6>0>p_2,\, p_4 $ and
$\alpha_2=\alpha_4=\alpha^*$ and $\alpha_6 =\alpha^{**}$.

Let  an  optimal cell $Y=\{y_1,\dots,y_6\}$ be given and define the
corresponding $\alpha_i$ and $p_i$. By possibly relabeling the atoms (in such a way that neighbors
remain neighbors) we can reduce ourselves to one of the following cases:  (1) $p_i\leq0$ for $i=2,4,6$ or (2)
$p_6\geq 0\geq p_2,p_4$. Note that these cases exhaust all possibilities, being
however not mutually exclusive. The statement follows now by checking that,  up to isometry,  
  $Y=Z$ in Case (1) and $Y=C$ in
Case (2).

Assume that we have $p_i\leq 0$, namely Case (1). Drop the constraint
$\theta_{3}=\theta^*$ by keeping all others (all bonds have length
$\ell^*$ and all bond angles other than $\theta_{3}$ are equal to
$\theta^*$). This uniquely defines
$\theta_{3}$ as a function of
$\alpha_6$, namely $\theta_3 =\theta_{3}(\alpha_6)$. Indeed, there
exists $\alpha^*<\alpha_{\rm max}^* <\pi$ such that for all
$\alpha_6\in [0,\alpha_{\rm max}^*]$ one
can uniquely determine $\alpha_2=\alpha_4 \in [0,\pi] $ with $\theta_{1}=
\theta_{5}=\theta^*$ by keeping $p_2,\,p_4\leq 0$ and for $\alpha_6>
\alpha_{\rm max}^*$ such values $\alpha_2,\,\alpha_4 $ do not exist.
Note that the mapping $\alpha_6\mapsto
\alpha_2=\alpha_4$ is strictly  decreasing.  Moreover,
$\alpha_2(\alpha^*) =\alpha_4(\alpha^*)=\alpha^*$. Indeed, if this was
not the case, the bond angles $\theta_1$ and $\theta_5$ would not be
$\theta^*$.  Corresponding to changes in
$\alpha_2=\alpha_4$ and for $p_2,\,p_4  \le  0$, the angle
$\theta_{3}$ changes as well and the mapping $\alpha_2=\alpha_4
\mapsto \theta_{3}$ is strictly   decreasing.  This entails that the
composed mapping $\alpha_6 \mapsto \theta_{3}(\alpha_6)$ is strictly
 increasing.    Hence, the equation $\theta_{3}(\alpha_6) = \theta^*$
has a unique solution. Such solution is necessarily 
  $\alpha_6=\alpha^*$, for this happens to be the case for
  $Z$. Recalling that 
$\alpha_2(\alpha^*) = \alpha_4(\alpha^*)  = \alpha^*$, we have
hence proved that  $\alpha_i =\alpha^*$ for
$i=2,4,6$,  so that $Y$   is necessarily $Z$. 

Assume now that $p_6\geq 0\geq p_2,p_4$, namely Case (2). Drop the constraint
$\theta_{3}=\theta^*$ by keeping all others. Let
$\alpha^{**}<\alpha_{\rm max}^{**} <\pi$ be given such that for all
$\alpha_6\in [0,\alpha_{\rm max}^{**} ]$ one
finds uniquely $\alpha_2=\alpha_4 \in [0,\pi] $ with $\theta_{1}=
\theta_{5}=\theta^*$ by keeping $p_2,\,p_4\leq 0$ and for $\alpha_6>
\alpha_{\rm max}^{**} $ such values $\alpha_2,\,\alpha_4 $ do not exist. Note that the mapping  $\alpha_6\mapsto
\alpha_2=\alpha_4$  is strictly increasing  and that
$\alpha_4(\alpha^{**})=\alpha_6(\alpha^{**})=\alpha^*$.   Indeed, if this was
not the case, the bond angles $\theta_1$ and $\theta_5$ would not be
$\theta^*$.  On the other hand, the mapping $\alpha_2=\alpha_4
\mapsto \theta_{3}$ is strictly decreasing.  Thus,  the  
composed mapping $\alpha_6 \mapsto \theta_{3}(\alpha_6)$ is strictly
decreasing and the equation $\theta_{3}(\alpha_6) = \theta^*$ has the
only solution $\alpha_6=\alpha^{**}$, for this corresponds to $C$. As $\alpha_2(\alpha^{**}) = \alpha_4(\alpha^{**})  = \alpha^*$, we have proved that $\alpha_2=\alpha_4 =\alpha^*$
and $\alpha_6=\alpha^{**}$.  In particular,  $Y$ is $C$. 
\end{proof}

\section{Classification of ground states}\label{sec:char}

Proposition \ref{CZ} provides a {\it local} description of
ground-state geometries. The purpose of this section is to move from
such  a  local description to the global picture. This is made
possible  as  $Z$ and $C$ cells  can be arranged in three-dimensional space just in few very
specific {\it global} patterns. This eventually allows us to classify 
ground-state deformations in Theorem \ref{th:characterization}.

  In order to state our result, we need to introduce some finer
description of cell geometries. Note indeed that Proposition \ref{CZ}
identifies optimal cells as point sets {\it up to isometries}. Here we need to
specialize this identification by taking into account the indicization of
the atoms as well. In particular, we say that two optimal cells
$\{y_1,\dots,y_6\}$ and $\{z_1,\dots,z_6\}$ are
{\it of the same type} if they are isomorphic via an isometry  $\iota: \Rz^3 \to \Rz^3$ 
with the property that $\iota(y_i)=z_i$  for $i=1,\ldots,6$. 

In order to find all possible {\it types} of
optimal cells,
one has to consider all permutations $\{i_1,\dots,i_6\}$ of the atomic
indices
$\{1,\dots,6\}$ which preserve first neighbors, namely such that $|i_k -
i_{k-1}|=1$ (the sum being modulo $6$). Such permutations are generated
by the two transformations $i \to i+1$ and $i \to -i$.

 As $Z$ cells as point sets are invariant under  $2\pi/3$ rotations about
their axis $n_0$ (see Figure \ref{CZfigure}) and are central symmetric, by
applying such generating transformations to the atomic indices  of $Z$
cells we identify exactly two equivalence classes: We say
that a $Z$ cell $\{y_1,\dots,y_6\}$ is of {\it type} $Z$ if it is
of the same type of the $Z$ cell of Figure \ref{CZfigure} and that it
is of {\it type} $\ove Z$ if it is
of the same type of the $Z$ cell of Figure \ref{CZfigure} up to
letting $y_i
\to y_{-i}$. Type $Z$ cells $\{y_1,\dots,y_6\}$ are
transformed into type $\ove Z$ cells (and viceversa) both by  $y_i \to
y_{i+1}$  or  $y_i \to y_{-i}$.

As $C$ cells are less symmetric than $Z$ cells, the type count for $C$
cells is necessarily higher.
A $C$ cell $\{y_1,\dots,y_6\}$ is said to
be of {\it type} $C$ if it is
of the same type of the $C$ cell of Figure \ref{CZfigure} and to be of
{\it type} $\ove C$ if it is
of the same type of the $C$ cell of Figure \ref{CZfigure}  up to
letting $y_i \to y_{-i}$. On the other hand, a $C$ cell is said to be
of {\it type}
$C_\pm$  ($\ove C_\pm$) if it is
of the same type of the $C$ cell of Figure \ref{CZfigure}  up to
letting $y_i \to y_{i\pm 1}$ ($y_i \to y_{-(i\pm 1)}$, respectively). 
 Type $C$ and $C_\pm$ cells $\{y_1,\dots,y_6\}$ are  respectively 
transformed into type $\ove C$ and $\ove C_\pm$ cells by the transformation $y_i \to y_{-i}$. 

The above provisions define a {\it type function}
$$ \tau: \Zz^2 \to \{Z,\ove Z, C, \ove C, C_+, \ove C_+, C_-, \ove C_-\}$$
which associates to each cell $(s,t)\in \Zz^2$ its type
$\tau(s,t)$. The cells in Figure \ref{CZfigure} are of type $Z$
(left) and type $C$ (right). A  type $\ove Z$ and  type $\ove C$ cell can be
visualized by taking the  reflection of a  type $Z$ and type $C$ cell, respectively,  with respect
to the plane $(e_1,e_3)$.


 Define  the {\it center}
$y_{\rm c}$ of the cell by
$$ y_{\rm c} = \frac14 (y_1+y_2+y_4+y_5).$$
To each
bond $\{y_{i},y_{i+1}\}$ we associate  a {\it bond
   plane} defined as the plane containing the endpoints of the bond
 and the center  $y_{\rm c}$  of the cell, oriented by the  unit
 vector  $n$ with direction 
 $(y_i{-}y_{\rm c})\wedge(y_{i+1}{-}y_{\rm c})$.  
 
 Let now the two cells $(s,t)$ and
 $(s',t')$ with centers $y_{\rm c}$ and $y_{\rm c}'$ share the bond
 $\{y_i,y_j\}$. We define the {\it
  signed incidence angle} $\gamma\in [-\pi,\pi]$ at the bond $\{y_i,y_j\}$ of the corresponding
bond planes  as
$$
\gamma = 
\left\{
\begin{array}{ll}
 \phantom{-} \arccos (n\cdot n') \quad&\text{if} \ \ (y'_{\rm c}{-}y_{\rm c})\cdot (n'{-}n)
                              <   0\\
 -\arccos (n\cdot n') \quad&\text{if} \ \ (y'_{\rm c}{-}y_{\rm c})\cdot (n'{-}n)
                              \geq  0
\end{array}
\right.
$$ 
 where $n$ and $n'$ denote the unit vectors to the bond planes in the cells $(s,t)$ and
 $(s',t')$, respectively.  Note that this definition is invariant under the transformation $(s,t)
\leftrightarrow (s',t')$ and that  $|\gamma|$ is the classical incidence angle between the two
bond planes.

In the following, an important role will be played by the angle
(recall  definitions  \eqref{h}) 
\begin{equation}\label{gammastar}
\gamma^* = 4\arctan(h/v) =  4\arctan \left(\frac{1+\cos\theta^* -
  2\cos^2\theta^*}{1-\cos\theta^*}\right)^{1/2}. 
\end{equation}
Note that $\gamma^*/2$ is the incidence angle of the two planes
containing  the atoms  $\{y_1,y_2,y_3,y_6\}$
 and $\{y_3,y_4,y_5,y_6\}$ of the type $C$ cell, see Figure
 \ref{CZfigure}. For each cell $(s,t)$, we let $\haz
\gamma(s,t)$ be the signed incidence angle at the common bond between
cell 
$(s,t)$ and  cell  $(s,t-1)$. 
 This notation allows us to state  our main result.

\begin{theorem}[Classification  of ground states]\label{th:characterization}
A deformation is a ground state if and only if, possibly up to a
reorientation
of  the reference lattice  $H$,  the type function $\tau$ takes values only in
$\{Z,\ove Z,C,\ove C\}$ and one of the following
two cases occurs
\begin{center}
  \begin{minipage}{0.85\linewidth}
    \vspace{1mm}

    (Zigzag roll-ups) \hspace{9.5mm} $\haz \gamma \equiv -\gamma^*$ and
    $\tau\equiv C$ or $\,\haz \gamma \equiv \gamma^*$ and
    $\tau\equiv  \ove C$. \vspace{2mm}

(Rippled structures)\qquad $\haz \gamma\equiv 0$ and   $t
    \mapsto   \tau(s,t)$   is constant for all $s \in \Zz$,

  \end{minipage}
\end{center}
\end{theorem} 

The fact that the type function 
 takes values exclusively in $\{Z,\ove Z,C,\ove C\}$ and, in
 particular,   values $C_\pm$ and
$\ove C_\pm$ do not occur  is due to  the   reorientation of the reference lattice
$H$. Even without such reorientation, a statement in the spirit of
Theorem \ref{th:characterization} would hold. The four possible values of the type
function would then be either $\{Z,\ove Z,C,\ove C\}$, $\{Z,\ove
Z,C_+,\ove C_+\}$, or $\{Z,\ove Z,C_-,\ove C_-\}$.

 The classification of Theorem \ref{th:characterization} says that  exactly two
 families of ground states exist.
In case $\haz \gamma \equiv  -\gamma^*$ and  $\tau  \equiv  C$ (or, equivalently, $\haz \gamma \equiv  \gamma^*$ and $\tau  \equiv 
\ove C$)  the ground state is   a rolled-up structure which is usually referred to as of {\it zigzag} type, see Figure
\ref{nanotube}. 
\begin{figure}[h]
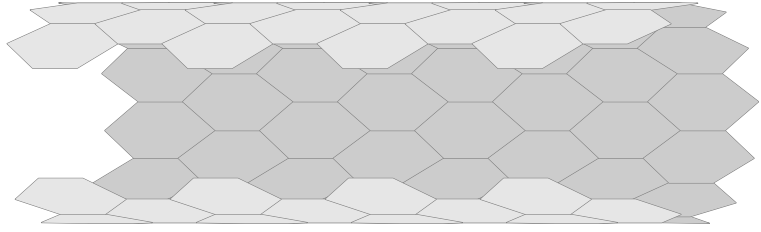

  \centering
  \pgfdeclareimage[width=100mm]{nanotube}{zigzag_nanotube} 
\pgfuseimage{nanotube}
\caption{Zigzag structure: $\tau \equiv C$ and $\haz \gamma \equiv
  -\gamma^*$ (or, equivalently, $\tau \equiv \ove C$ and $\haz \gamma \equiv
  \gamma^*$).}\label{nanotube}
\end{figure}  
 If $\gamma^*= \pi/m$ for some $m \in \Nz$
large enough,  the configuration is a {\it zigzag nanotube} with $m$ cells on each
section.  Note, however, that such condition on $\gamma^*$ is nongeneric
with respect to the choice of the  energy,  see \eqref{gammastar}.
The zigzag ground-state deformation  is  not injective iff
$k\gamma^* = \pi/m + 2\pi j$ for some $k,\, j \in \Nz$.

The second possibility from the  classification of Theorem
\ref{th:characterization}   is that $\haz \gamma  \equiv  0$. In  this case,  
the ground state corresponds to an alternation of
cell types which are constant along the coordinate direction $b$. The
ground state is hence uniquely determined by the sequence of types,
e.g. $\{\dots, C,C,\ove Z, Z,C,\ove C,\dots\}$. All such sequences can
in principle be considered, although some of them give rise to
noninjective deformations or even self-interpenetrating
structures.

The choices
\begin{align*}
 &\{\dots,C,\ove C,C,\ove C, C\dots\},\\
&\{\dots,C,C,\ove C,\ove C,C,C,\ove C,\ove C, C,C\dots\},\\
&\{\dots,C,C,C,\ove C,\ove C,\ove C,C,C,C,\ove C,\ove C,\ove C,C,C, C\dots\},
\end{align*}
originate {\it ripples} with different wave lengths, corresponding to
the different number of copies of $C$ and $\ove C$ in the
sequence. Choices including $Z$ and $\ove Z$ cells can generate
ripples as well, see Figure \ref{pink}.
\begin{figure}[h]
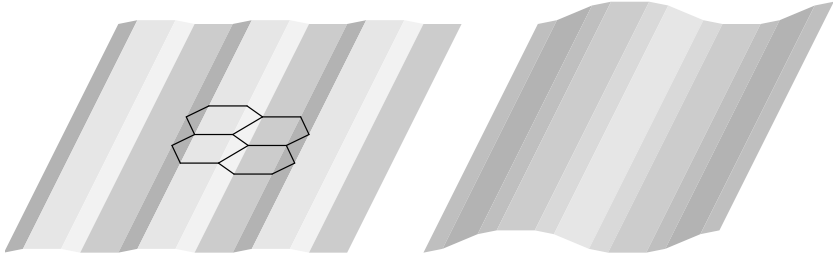

  \centering
  \pgfdeclareimage[width=110mm]{pink}{pinkfluffyunicorn} 
\pgfuseimage{pink}
\caption{Examples of rippled structures for $\haz \gamma\equiv 0$: $\{\dots,C,\ove C,C,\ove C,
  C\dots\}$ (left) and $\{\dots,  \ove Z,  \ove C, Z,  C, \dots\}$ (right).}\label{pink}
\end{figure}

The constant choice $\{\dots,C,C,C,\dots\}$ gives
rise to a rolled-up structure of the so-called {\it armchair} type,
see Figure \ref{pink2}. 
\begin{figure}[h]
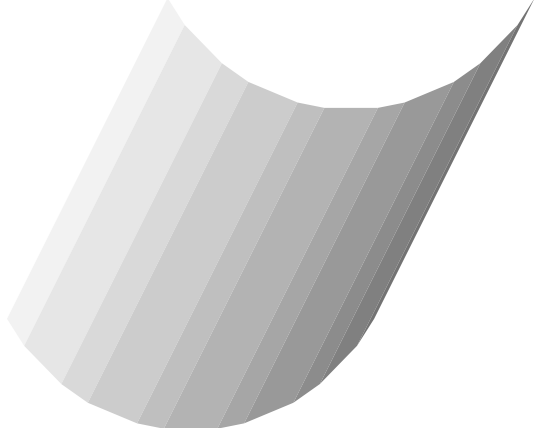

  \centering
  \pgfdeclareimage[width=70mm]{pink2}{pinkfluffyunicorn2} 
\pgfuseimage{pink2}
\caption{Armchair structure: $\{\dots,C,C,C, \dots\}$ for
  $\haz \gamma\equiv 0$.}\label{pink2}
\end{figure}
If one has that 
$$ \beta=\arctan \left( \frac{h}{\ell - \ell^*/2}\right) =
 \frac{\pi}{m}  $$ for some $m \in \Nz$ large enough (see Figure
\ref{CZfigure} bottom right), the
rippled structure $\{\dots,C,C,C,\dots\}$ closes up and we have an
{\it armchair nanotube}  \cite{tube, numeric-stability}  with $m$ atoms on each (nonempty)
section.  Again, the condition on $\beta$ is nongeneric.

In the rippled case, ground states are essentially one dimensional. Indeed,
the sequence of cell types  is  completely characterized by any section with respect
to direction $b^\perp$ in $H$, see Figure \ref{sections}. One can hence introduce an effective energy for such sections by
considering cell centers as particles and favoring a specific
distance  between cell centers  and a specific angle  $\varphi^*$ between segments connecting
neighboring cell centers. We follow this path in \cite{Friedrich17}
where we show that third-neighbor interactions between cell centers
and  certain boundary  conditions select
specific optimal ripple lengths, independently of the  sample size 
(assumed to be sufficiently large). 

Note that for all  $\varphi$  close to $\pi$ one
can find $\theta =  \theta( \varphi)   $ close to $2\pi/3$ so that, by letting all bond angles
of the $C$ and
$\ove C$ cells in
Figure \ref{sections} be $\theta$ (possibly being not optimal), 
the 
segments connecting cell
centers form  $\varphi$  angles
which each other. The ground state corresponds then to
$\theta=\theta^*$ or  $\varphi=\varphi^*$  and one can check that
$$ 2\pi/3 - \theta^* \sim (\pi -  \varphi^* )^2.$$
In particular, by defining $\haz v_3( \varphi ) = v_3( \theta( \varphi )  )$ (and letting $v_3$
be smooth in $2\pi/3$) one has that $\haz v_3$ is minimized in $\pi$
with 
\begin{equation}\label{forpaper2}
\haz v_3'(\pi) = \haz v_3''(\pi) = \haz v_3'''(\pi) = 0\quad
\text{and} \quad  \haz v_3''''(\pi)>0.
\end{equation}

\begin{figure}[h]
  \centering
  \pgfdeclareimage[width=130mm]{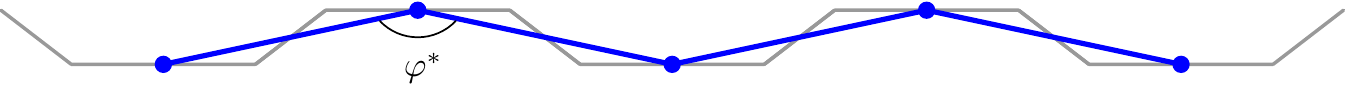}{sections} 
\pgfuseimage{sections}
\caption{Effective description of a $b^\perp$ section of the rippled
  structure  $\{\dots,C,\ove C,C,\ove C,\dots\}$.}\label{sections}
\end{figure}

\begin{proof}[\bf Proof of Theorem \ref{th:characterization}]
  The argument is combinatorial in nature and follows by investigating all possible
  cases. The main idea is that cells of different type sharing a bond
  have a limited number of possible mutual arrangements. We start by
 discussing the local geometry of bonds in Step 1 and turn to 
 arrangements of two cells in Step 2.  The reorientation of the
 reference lattice $H$ is described in Step 3.  The case of three or more cells
 is discussed in  Steps 4-7. Finally, in Step 8 we conclude that only rippled structures and zigzag roll-ups are admissible. 

\noindent{\bf Step 1: Defining the bond type.} 
We start by introducing some notation for the
  various bonds of the different types of cells. By referring to the notation
  of Figure \ref{CCZZfigure}, we have that  the 
atoms $y_1,\,y_2,\,
y_4$, and $y_5$ of each cell are coplanar. The shadings in Figure \ref{CCZZfigure} allude to the fact that
the cells are indeed not flat and the signs 
$+$ and $-$  
illustrate the positioning of $y_3$ and $y_6$ with respect
to the plane containing $y_1,\,y_2,\,
y_4$, and $y_5$.

\begin{figure}[h]
  \centering
  \pgfdeclareimage[width=110mm]{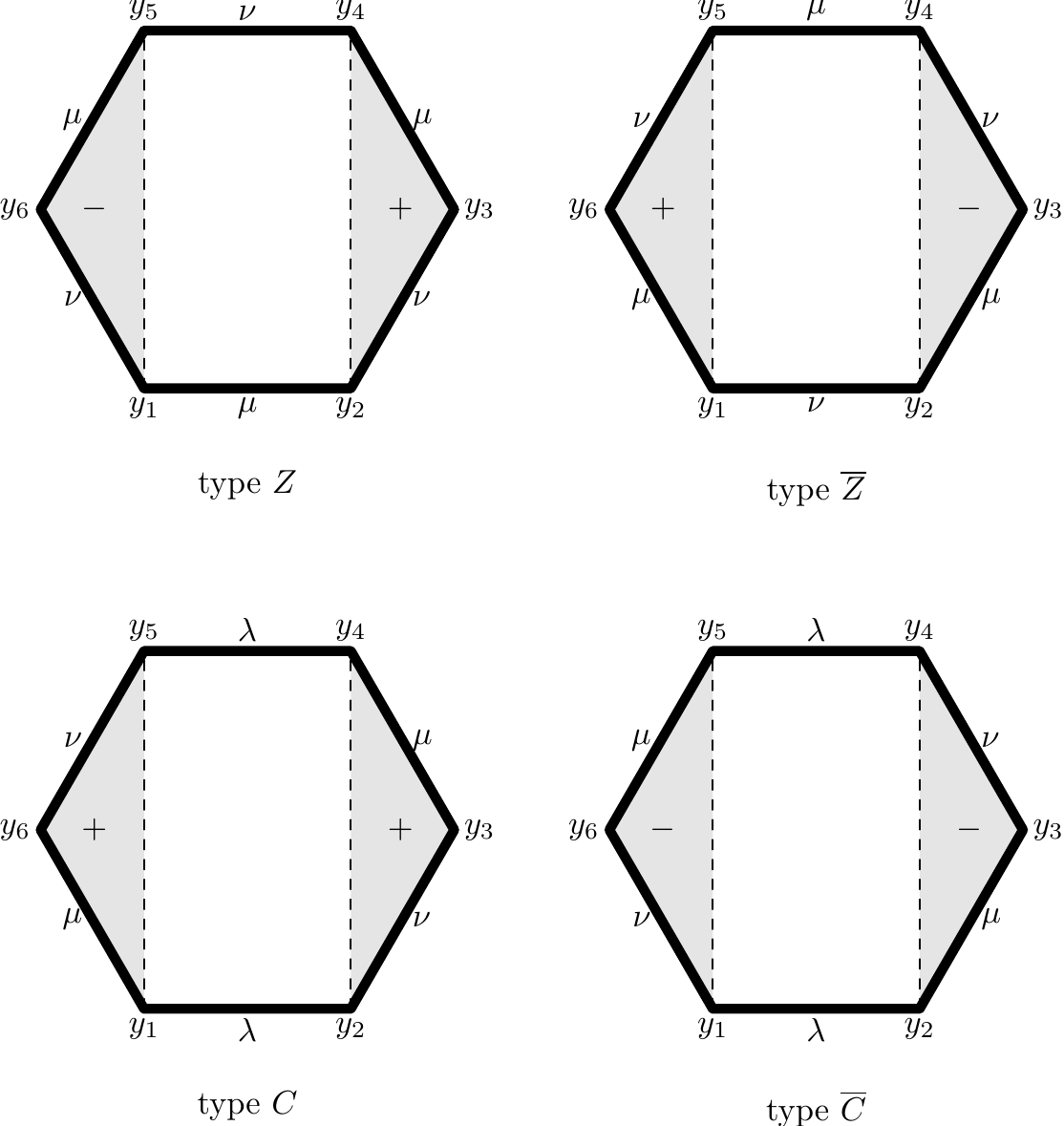}{CCZZcells} 
\pgfuseimage{CCZZcells}
\caption{Notation for the bonds of different cell types  (top
  view). Cells of type $C_\pm$ and $\ove C_\pm$ are not
  illustrated as they are excluded by the reorientation of $H$ of Step
  3.  }\label{CCZZfigure}
\end{figure}

Given the bond $\{y_{i},y_{i+1}\}$ recall that the bond plane containing
$y_{i}$, $y_{i+1}$, and $y_{\rm c}$ is oriented via the unit vector  $n$  with
direction $(y_i{-}y_{\rm c})\wedge(y_{i+1}{-}y_{\rm c})$ and define $u
= (y_{i+2}{-}y_{i-1}){\cdot}n$.
We say that the bond  $\{y_{i},y_{i+1}\}$ is of {\it type} $\lambda$ if
$u=0$, of type $\mu$ if $u>0$, and of type $\nu$ if $u<0$. Note that
the bond $\{y_{i},y_{i+1}\}$ is of type $\lambda$ iff the four atoms
$\{y_{i-1}, y_i, y_{i+1}, y_{i+2}\}$ are coplanar. 

This distinction of
bond types will turn out useful for discussing mutual cell
 arrangements.  In particular, we say that {\it two cells share a $(\mu,\nu)$ bond} if  the common
 bond for such two cells is of type $\mu$ for one cell and of type
 $\nu$ for the other. Analogously for $(\lambda,\lambda)$ bonds,
 $(\lambda,\mu)$ bonds etc.

\noindent{\bf Step 2: Sharing a bond.} The aim of this step is to
classify the possible mutual arrangements of two cells sharing a
bond. This is specified in terms of a corresponding signed incidence
angle. By considering the bond angles at the endpoints of the shared
bonds which are external to the cells (named {\it external}
henceforth), we find the admissible values of the signed incidence angles. All possibilities are listed in Table \ref{angletable}
below. We now comment on its entries.

\begin{figure}[h]
  \centering
  \pgfdeclareimage[width=140mm]{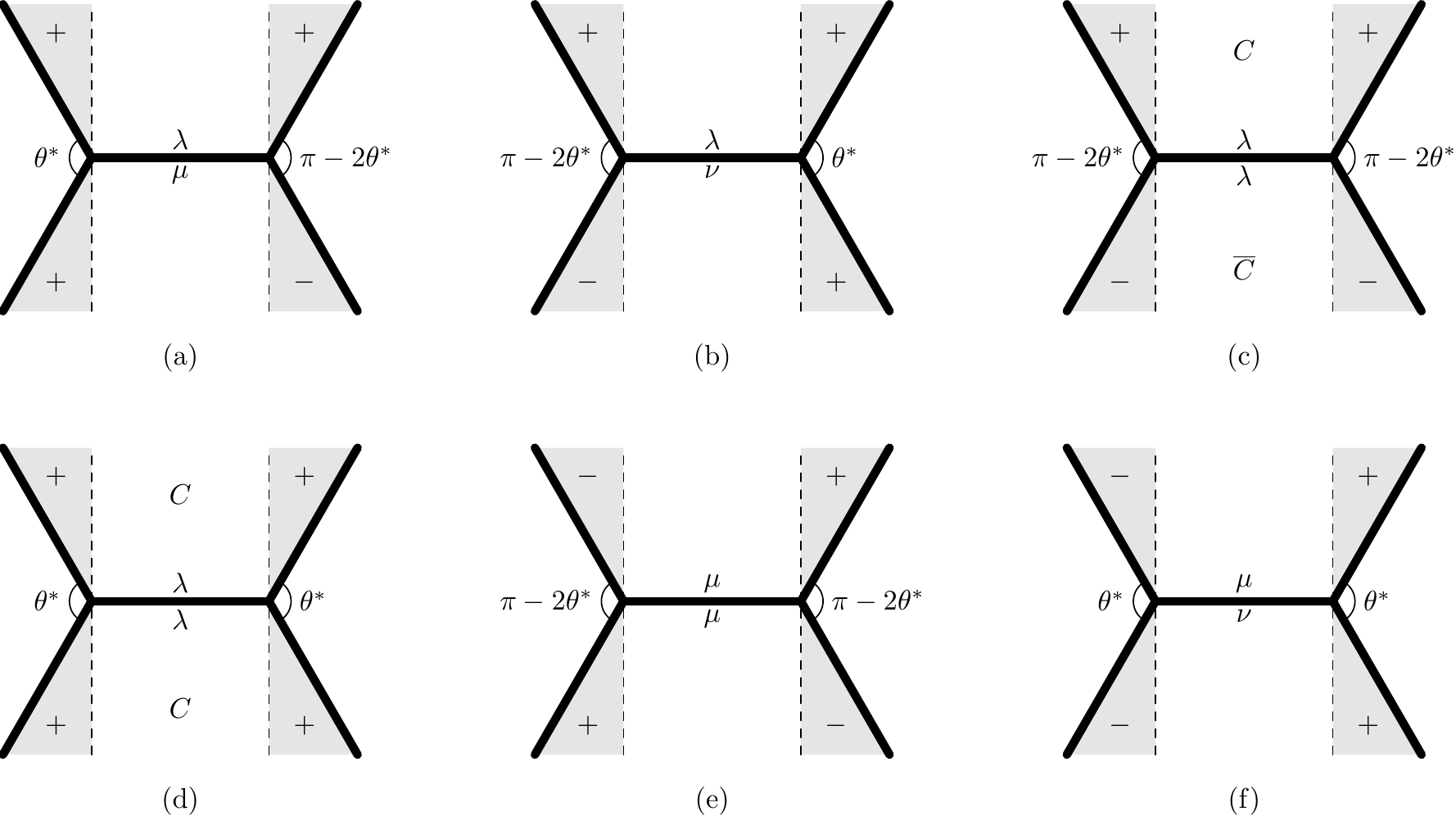}{CZbonds} 
\pgfuseimage{CZbonds}

\caption{Illustration of the various types of shared bonds for
  $\gamma=0$.  If a triangle with sign $+$ and a triangle with sign
  $-$ share an atom,  the three bonds including such atom  are coplanar. Therefore, the external angle is $\pi - 2 \theta^*$. If two triangles with equal sign share an atom, by elementary   
trigonometry the external angle is $\theta^*$, cf. \eqref{h}. 
These cases correspond to those of Figure \ref{CCZZfigure}. In
particular, the cases corresponding to cell types $C_\pm$ and $\ove C_\pm$ are not illustrated as
they are excluded by the reorientation of $H$ of Step 3. }\label{CZbonds}
\end{figure}

\begin{table}[h]
\begin{tabular}{|c|c|c|}
\hline
Type of shared bond & Signed incidence angle & Reference in Figure \ref{CZbonds}\\
\hline
$(\lambda,\mu)$, $(\lambda,\nu)$ & $\not \exists$& (a), (b)\\
\hline
$(\lambda,\lambda)$ ($C$ and $\ove C$)& $ \pm
\gamma^*/2$ (not admissible)& (c)\\
$(\lambda,\lambda)$ (two $C$ cells)& $0 ,\, - 
\gamma^* $& (d) \\
$(\lambda,\lambda)$ (two $\ove C$ cells)& $0 ,\,  
\gamma^*  $& (d)\\ 
\hline
$(\mu,\mu)$, $(\nu,\nu)$ & $\pm \gamma^*/2 $  &(e)\\
\hline
$(\mu,\nu)$&$0$&(f)\\
\hline
\end{tabular}
\caption{Signed incidence angles for various types of shared
  bonds. Note that the case $(\lambda,\lambda)$ for a $C$ and a $\ove
  C$ cell    will be
  eventually proved to be not admissible in Step 7.}\label{angletable}
\end{table}

 \emph{Cases $(\lambda,\mu)$ and $(\lambda,\nu)$}:  As the two atoms  at a 
$\lambda$ bond and their first neighbors are coplanar, by referring to
Figures \ref{CZbonds}(a) and \ref{CZbonds}(b) one realizes that the
external angles for a 
$(\lambda,\mu)$ or $(\lambda,\nu)$ bond cannot be both $\theta^*$, for
any $\gamma$. As a consequence, two cells cannot share 
  a $(\lambda,\mu)$ bond nor a $(\lambda,\nu)$ bond.

\emph{Case $(\lambda,\lambda)$}: Assume that two cells
share a $(\lambda,\lambda)$ bond and let $\gamma$ be the signed
incidence angle formed by the corresponding bond planes. As $Z$ and
$\ove Z$ cells do not have $\lambda$ bonds, see Figure
\ref{CCZZfigure}, the cells sharing the $(\lambda,\lambda)$ bond are
necessarily $C$ or $\ove C$. Let a $C$ and a $\ove C$ cell share
$(\lambda,\lambda)$ bond. By referring to  Figure \ref{CZbonds}(c) one
realizes that the incidence angle $\gamma$ at the shared bond cannot
be $0$, for this would imply that the external angles  are
$\pi - 2 \theta^*>\theta^*$. Due to symmetry, one finds exactly two
symmetric values of the incidence angle ensuring such external angles
to be  $\theta^*$. In particular, we have that $\gamma=\pm
\gamma^*/2$, where $\gamma^*$ is defined in  \eqref{gammastar}.
 The occurrence of a $(\lambda,\lambda)$
bond between a $C$ and a $\ove C$ cell will be however proved to be not
admissible in Step 7. 
If both cells
are of type $C$, the signed incidence angle $\gamma$ is either
$0$, see Figure \ref{CZbonds}(d),  or $-\gamma^*$, for these are the
only values ensuring that the external angles  are $\theta^*$. By symmetry, in case both cells
are of type $\ove C$, the signed incidence angle is either $0$ or $\gamma^*$.

\emph{Cases $(\mu,\mu)$ and $(\nu,\nu)$}: By varying the signed incidence angle $\gamma$, the
external angles remain equal and are strictly decreasing with respect to
$|\gamma|$. As such external angles are  $\pi - 2\theta^*>\theta^*$  for
$\gamma=0$, see Figure \ref{CZbonds}(e) for the case $(\mu,\mu)$, one finds exactly two symmetric values of $\gamma$ making
them equal to $\theta^*$. By referring to the discussion of the
$(\lambda,\lambda)$ bond, one can check that such values are
exactly $\pm \gamma^*/2$.  

\emph{Case $(\mu,\nu)$}: The external angles are $\theta^*$ for $\gamma=0$,
see Figure \ref{CZbonds}(f) and are antimonotone with respect to
$|\gamma|$. As such, $\gamma=0$ is the only admissible value for the
incidence angle.

\noindent {\bf Step 3: Reorienting $H$.} Given a ground state, we show in this step that one
can reorient the reference lattice $H$ in such a way that only the
cell types $\{Z,\ove Z, C, \ove C\}$ occur.

Assume that the ground state
contains a $C$ cell. By reorienting $H$ one can assume it to be 
of type $C$  or $\ove C$.  Letting such cell be indexed by $(s_0,t_0)$ we have that
cells $(s_0,t_0\pm 1)$ are necessarily either of  type $C$
or $\ove C$, for they all need to share a $(\lambda,\lambda)$ bond
with $(s_0,t_0)$. By iterating the argument we have that all cells
$(s_0,t)$, for $t\in \Zz$, are either of  type $C$
or $\ove C$.
 We now prove that the ground state contains no type
$C_\pm$ nor $\ove C_\pm$ cells. Assume indeed that  cell $(s_1,t_1)$
is of type  $C_-$  (analogously for  $C_+$  and $\ove C_\pm$). Then, the
same argument as above entails that all cells $(r,t_1)$, for $r\in \Zz$, are
either of type  $C_-$ or $\ove C_-$.  This, however, brings to a
contradiction as cell $(s_0,t_1)$ would have to be both of type $C_-$ or $\ove C_-$  and $C$ or $\ove C$.

Having fixed the orientation
of $H$, all $C$ cells are of type
$C$ or $\ove C$, so that we just refer to $C$ and $\ove C$ cells   in
the following, omitting the word {\it type}. Note that for each $C$
and $\ove C$ cell the $\lambda$
bonds are $\{y_1,y_2\}$ and $\{y_4,y_5\}$.
If the ground state contains just type $Z$ and $\ove Z$ cells, no
reorientation of $H$ is actually needed. In all cases, by considering
arrangements of
cells of a ground state we can always refer to the
orientations of Figure \ref{CCZZfigure}.

\noindent {\bf Step 4: Rippled structures, special case.}  
Let  us  start by considering the special case of 
  two $C$ cells sharing a $(\lambda,\lambda)$ bond with $\gamma=0$. The goal is here
to show that two neighboring cells of such $C$ cells must be of the same type (not
necessarily $C$) and share a bond with $\gamma=0$. This fact will be
used in an induction argument in Step 5.

We can assume with no loss of generality that the joined $C$
cells are $(0,0)$ and $(0,1)$. 
We proceed by discussing cases.

 \emph{Case $ \tau(1,0)= C$}:  One has that the shared bond between
$(0,0)$ and $(1,0)$ is
of type $(\mu,\mu)$. We directly
check that 
 $\tau(1,1)\not \in\{ Z,\, \ove Z\}$ because in this case the cells
 $(1,0)$ and $(1,1)$ would share a $(\lambda,\mu)$ or a
 $(\lambda,\nu)$ bond, which is not admissible, see Table
 \ref{angletable}. The case $\tau(1,1)=\ove C$ is also excluded: The
 $(\mu,\nu)$ bond shared by cells $(0,1)$ and $(1,1)$ requires the
 corresponding signed incidence angle to be $0$, see Table
 \ref{angletable}, and the two cells $(1,0)$ and $(1,1)$ would have
 no shared bond. Indeed, by referring to the notation of Figure
 \ref{noshared}(a), one has that  the three bonds between the cells  $(0,0)$, $(1,0)$, the cells $(1,0)$, $(0,1)$, and the cells  $(0,1)$, $(1,1)$ are coplanar,  the atoms in the darkened regions belong
 to two parallel planes, whereas atoms $\{y_1,y_2,y_4,y_5\}$ of
 cell $(1,0)$ are not
 coplanar with those of cell $(1,1)$. As such, the marked bond cannot
 be shared by cells $(1,0)$ and $(1,1)$.
\begin{figure}[h]
  \centering
  \pgfdeclareimage[width=130mm]{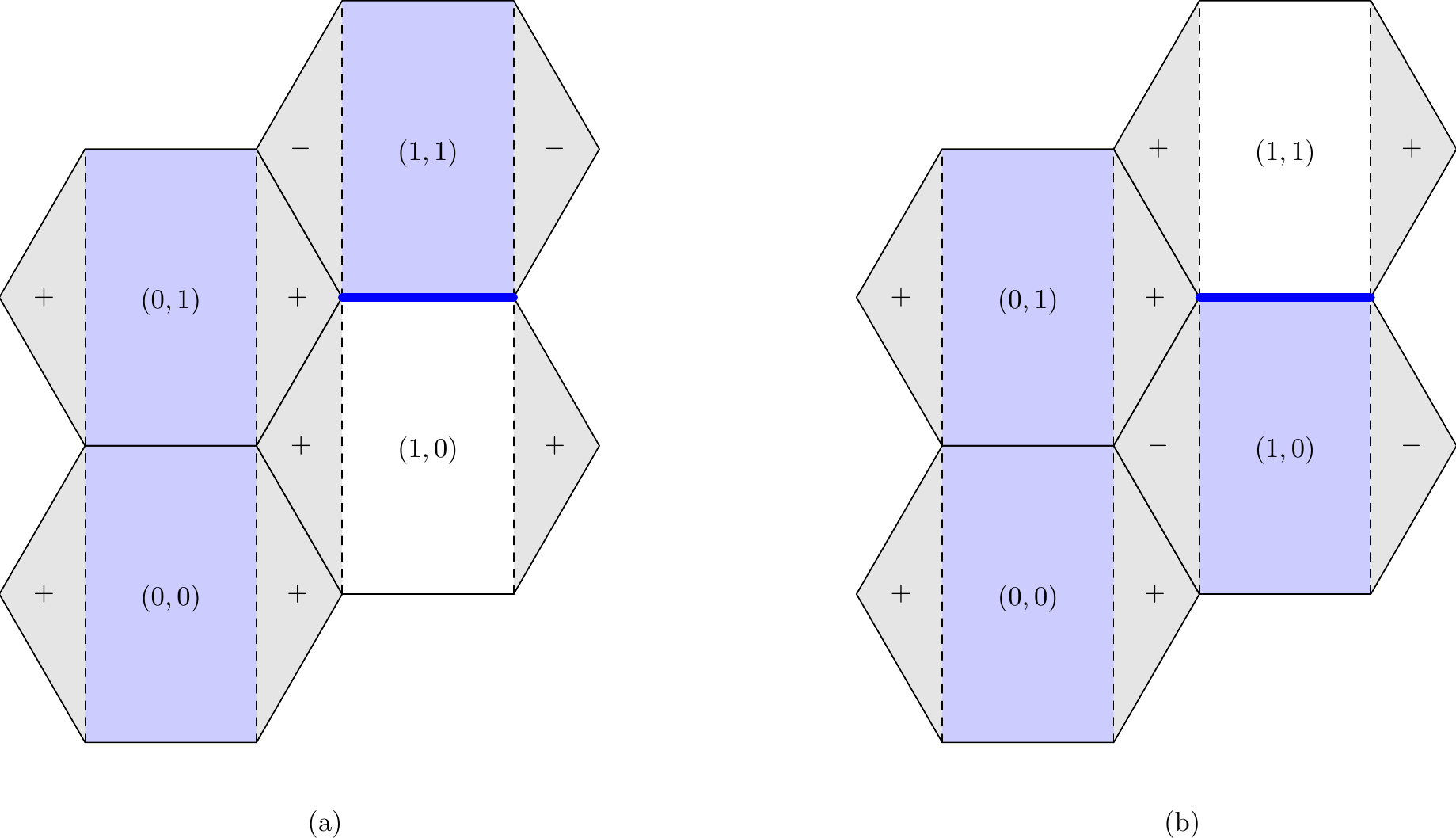}{noshared} 
\pgfuseimage{noshared}
\caption{Two nonrealizable configurations.}\label{noshared}
\end{figure}
The only possibility left is $ \tau(1,1)=C$, which can indeed be realized by
letting the signed incidence angle along the shared bond between cells $(1,0)$ and $(1,1)$ be $0$.

 \emph{Case $ \tau(1,0) = \ove  C$}:  One can again check that
$\tau(1,1)\not \in\{ Z,\, \ove Z\}$ because in this case the cells
 $(1,0)$ and $(1,1)$ would share a $(\lambda,\mu)$ or a
 $(\lambda,\nu)$ bond, which is not admissible. Moreover, the case $\tau(1,1)=C$ can be
excluded arguing similarly as above: By referring to Figure
\ref{noshared}(b) one has that  the three bonds between the cells  $(0,0)$, $(1,0)$, the cells $(1,0)$, $(0,1)$, and the cells  $(0,1)$, $(1,1)$ are coplanar,  the atoms in the darkened regions belong
 to two parallel planes, whereas atoms $\{y_1,y_2,y_4,y_5\}$ of
 cell $(1,0)$ are not
 coplanar with those of cell $(1,1)$. In particular,  the marked bond cannot
 be shared by cells $(1,0)$ and $(1,1)$. We are left with the possibility
 of having $\tau(1,1)=\ove C$, which can indeed be realized   by
letting the signed incidence angle along the shared bond between cells
$(1,0)$ and $(1,1)$ be $0$.

 \emph{Case $\tau(1,0)=Z$}:  One can argue exactly as in the case of $\tau(1,0)=\ove
C$ and find that $\tau(1,1)=Z$ as well, with a signed incidence angle
 along the shared bond between $(1,0)$ and $(1,1)$ being $0$. 
 Indeed, one can still refer to Figure \ref{noshared}(b) by forgetting the two
right-most atoms. 
 
 \emph{Case $\tau(1,0)=\ove Z$}:  One can argue exactly as in the case of $\tau(1,0)=C$  and find that 
$\tau(1,1)=\ove Z$ as well, with a signed incidence angle  along the shared bond
between $(1,0)$ and $(1,1)$ being $0$.  This case corresponds to  Figure \ref{noshared}(a) upon forgetting the two
right-most atoms.

In conclusion, within this step we have proved the following
\begin{align}
  &\tau(0,0)=\tau(0,1)=C \ \ \text{and} \ \ \gamma=0 \ \ \text{along the
  shared $(\lambda,\lambda)$ bond} \nonumber\\
&\Longrightarrow \quad 
\tau(1,0)=\tau(1,1) \ \ \text{and} \ \ \gamma=0 \ \ \text{along the
  shared bond}.\label{step3}
\end{align}

\noindent {\bf Step 5: Rippled structures, general case.} The
argument of Step 4 is purely based on bond types. As such, it can be
verbatim extended  to case $\tau(0,0)=\tau(0,1)\in\{Z,\ove
Z,\ove C\}$ as long as $ \gamma=0 $ along their
  shared bond. 
In addition, conclusion \eqref{step3}
  can be extended by symmetry to cell  $(1,-1)$ and cells $(-1,0)$, $(-1,1)$, and $(-1,2)$  as
  well. We hence have the following
\begin{align}
  &\tau(0,0)=\tau(0,1) \ \ \text{and} \ \ \gamma=0 \ \ \text{along the
  shared bond}\nonumber\\
& \Longrightarrow \quad
 \tau(-1,0) = \tau(-1,1) = \tau(-1,2) \ \  \   \text{and}  \ \ \  \tau(1,-1) = \tau(1,0) =
\tau(1,1)  \ \ \nonumber\\
&\text{and} \ \ \gamma=0 \ \ \text{along the
  shared bonds}.\label{step}
\end{align}

We can now use \eqref{step} iteratively and prove that if $\tau(0,0) = 
\tau(0,1)$ with  $\gamma=0$ along the shared bond. 
then
$\tau(s,\cdot)$ is constant for all $s \in \Zz$ and $\haz \gamma (s,t)  =0$ for all
$(s,t)\in \Zz^2$. Note that all cell types  $\{Z,\ove Z,C,\ove C\}$
 are admissible for
$\tau(s,\cdot)$.

This  proves the
Theorem in case $\tau(0,0) = 
\tau(0,1)$ with  $\gamma=0$ along the shared bond. 

\noindent {\bf Step 6: Zigzag roll-ups.} 
Let us now consider the case of  two $C$ cells sharing a
$(\lambda,\lambda)$ bond with $\gamma=-\gamma^*$. The goal is here
to show that  $\tau\equiv C$ and $\haz \gamma\equiv  - \gamma^*  $. 

  As in Step 4, assume with no
loss of generality that the $C$ cells are $(0,0)$ and $(0,1)$, namely
$\tau(0,0)=\tau(0,1)=C$. We aim at proving that
$\tau(1,0)=\tau(1,1)=C$ as well, 
which would imply that the signed incidence angle of the shared  $
(\lambda,\lambda)$  bond between $(1,0)$ and $(1,1)$ is again $
-\gamma^*$.

 \emph{Case $\tau(1,0) \in \lbrace \ove C, Z \rbrace$ (not admissible)}:   If this was the case,  cell $(1,0)$ would share a $(\mu,\nu)$ bond with
cell $(0,0)$. According to Table \ref{angletable}, the two
corresponding signed incidence angles  for $\ove C$ and $Z$, respectively,  would be $0$. This in
particular entails that the atoms $\{y_2,y_3,y_4\}$ of cell $(0,0)$
and $\{y_1,y_5,y_6\}$ of cell  $(1,0)$  have to be coplanar. At the
same time, cell $(1,0)$ would share a $(\mu,\nu)$ bond with cell
$(0,1)$ and the  atoms  $\{y_2,y_3,y_4\}$ of cell $(0,1)$
and $\{y_1,y_5,y_6\}$ of cell  $(1,0)$  would have to be
coplanar.  This
is however   impossible as the atoms $\{y_2,y_3,y_4\}$ in the two
cells  $(0,0)$ and $(0,1)$ are not coplanar, due to the condition
$\gamma = - \gamma^*$ along the shared bond between cells $(0,0)$ and $(0,1)$. 

 \emph{Case  $\tau(1,0)=\ove Z$ (not admissible)}:  Assume
that this  was  the case and consider cell $(1,1)$. This cannot be of type $C$ nor
$\ove C$, for in this case cell $(1,1)$  would share a $(\lambda,\mu)$  
bond (not admissible by Table \ref{angletable})   with cell  $(1,0)$.  On the other hand, cell $(1,1)$ cannot be of type $\ove
Z$ as in this case it would share a $(\mu,\nu)$ bond with cell $(1,0)$
and the corresponding  signed incidence angle 
$0$. We could then apply Step 5 in order to find that the signed incidence angle between cell
$(0,0)$ and $(0,1)$ would have to be
$0$ as well, which is a contradiction. 
The last possibility is that cell $(1,1)$ is of type $Z$. In this case,  cell
$(1,1)$ and cell $(0,1)$ share a $(\mu,\nu)$ bond and thus the
signed incidence angle along the shared bond  is $0$, see Table \ref{angletable}.  Similarly to the case of Figure \ref{noshared}(a),
cell $(1,1)$ would not share a bond with cell $(1,0)$.

 \textit{Case  $\tau(1,0)=C$}:   We have hence proved that, given  $\tau(0,0)=\tau(0,1)=C$  with
signed incidence angle   $ -\gamma^*$ along the shared  $(\lambda,\lambda)$
bond, the only possible type of cell $(1,0)$ is $C$. 
Cells $(1,1)$ and $(1,-1)$ need
then to be of type $C$ or $\ove C$ as well, for they have to share a
$\lambda$ bond with cell  $(1,0)$.  One can however exclude that they
are of type $\ove C$ since in this case the signed incidence angle to
cell $(0,1)$ or cell $(0,0)$,  respectively,   would be $0$ and they would
not share a bond with cell  $(1,0)$.    We again refer to Figure \ref{noshared}(a) for a similar argument. 

In conclusion, if   $\tau(0,0)=\tau(0,1)=C$  with
signed incidence angle   $ -\gamma^*$ along the shared 
bond, one has that  $\tau(1,1)=\tau(1,0)=\tau(1,-1)=C$.  This can indeed be realized by
letting the signed incidence angle along the shared bond between cells $(1,1)$, $(1,0)$ and $(1,0)$, $(1,-1)$  be $- \gamma^*$.  By symmetry, the same holds
for cells  $(-1,2)$, $(-1,1)$, and $(-1,0)$  as well. It is now easy to proceed by induction in order to prove that indeed
 $\tau(s,t)=C$  and $\haz \gamma (s,t)=-\gamma^*$ for all
 $(s,t)\in \Zz^2$.  

An analogous conclusion obviously holds in case
 $\tau(0,0)=\tau(0,1)= \ove C$ with signed incidence angle  $
\gamma^*$. In this case, $\tau(s,t)=\ove C$ and $\haz \gamma (s,t)=\gamma^*$ for all
 $(s,t)\in \Zz^2$.  
This  proves the
Theorem in case $\tau(0,0) = 
\tau(0,1) \in \{C,\, \ove C\}$ with  $\gamma=\mp \gamma^*$ along the shared bond.

\noindent {\bf Step 7:  Nonadmissible  configurations
  containing $C$ and $\ove C$ cells.}
In order to conclude the proof of the Theorem, one needs to check
that no other configurations of optimal cells  are  admissible  but those
already considered in Steps 5 and 6. This is done here and in Step 8.

If a ground state contains a $C$ or a $\ove C$ cell, it contains
infinitely many as these are the only ones that can share $\lambda$
bonds. Assume that cell $(0,0)$ is of type $C$. 
Then, all cells
$(0,t)$ are either $C$ or $\ove C$, see Step 3.  If two adjacent cells $(0,t)$ are
of the same type, one has that $t \mapsto \tau(0,t)$ is constant, due to Step
4. One is then left with the possibility that $\tau(0,t)=C$ for $t$ even and
$\tau(0,t)=\ove C$ for $t$ odd. The rest of the step is aimed at
proving that such an  alternation  of types is not  admissible.  

Let us start by checking that a configuration with
\begin{equation}
\tau(s,t) =
    C \ \ \text{if $s+t$ is even},\quad \tau(s,t) = 
\ove C  \ \ \text{if $s+t$ is odd,  for $s=0,1$,} \label{even}
\end{equation} 
is not admissible.  Indeed, in this case
the four coplanar atoms of cell $(0,0)$ and those of cell $(1,0)$
belong to parallel planes and atoms $\{y_2,y_3,y_4\}$ of cell $(0,0)$
and $\{y_1,y_5,y_6\}$ of cell $(1,0)$ are coplanar, see the darkened
region in  Figure \ref{noshared2}. At the same time, the four coplanar atoms of cell $(0,1)$ and those of cell $(1,1)$
belong to parallel planes and atoms $\{y_2,y_3,y_4\}$ of cell $(0,1)$
and $\{y_1,y_5,y_6\}$ of cell $(1,1)$ are coplanar. This, however, excludes
  that cells  
$(0,0)$, $(1,0)$ and cells $(0,1)$, $(1,1)$ simultaneously share the
three marked bonds in Figure \ref{noshared2} and configuration
\eqref{even} is not  admissible.  By symmetry, the configuration
\begin{equation}
\tau(s,t) =
    C \ \ \text{if  $t$  is even},\quad \tau(s,t) = 
\ove C  \ \ \text{if   $t$   is odd,  for $s=0,1$,} \label{odd}
\end{equation}
is not   admissible  as well.

\begin{figure}[h]
  \centering
  \pgfdeclareimage[width=57mm]{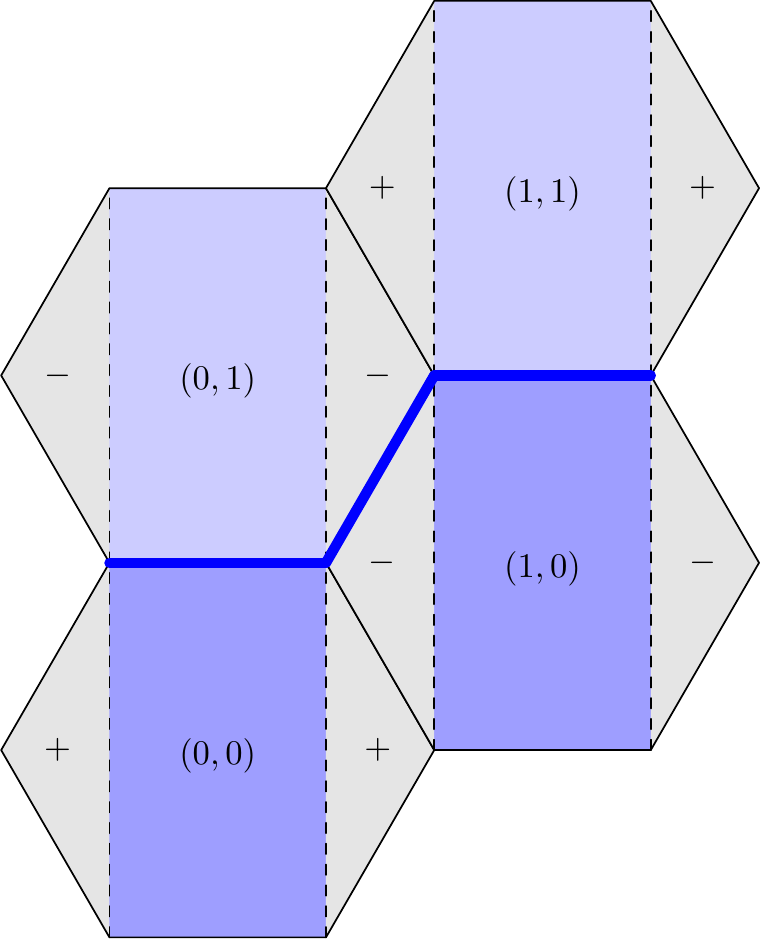}{noshared2} 
\pgfuseimage{noshared2}
\caption{Configuration \eqref{even} is
  not  admissible.}\label{noshared2}
\end{figure}

Assume now that $\tau(0,t) = C$ for $t$ even and $\tau (0,t)=\ove C$ for $t$
odd and $\tau(1,t) \in \{Z, \ove Z\}$. We can assume that neighboring cells $(1,t)$
are of different types  since   otherwise one would have a signed
incidence angle $0$ along a shared bond  (see Figure \ref{CCZZfigure} and Table \ref{angletable})  and we would be in the
situation of Step 5,  see \eqref{step}.  If $\tau(1,0)=Z$, we can argue exactly in the
case of \eqref{even} (by forgetting the two
right-most atoms in Figure \ref{noshared2}) and find that the
configuration is not  admissible.  Analogously, the case
$ \tau(1,0)  =\ove Z$ can be excluded by arguing as for
\eqref{odd}.

\noindent {\bf Step 8: Conclusion of the proof.} Let us now check that the previous steps exhaust all possible cases and that the statement holds.

If the ground state contains a $C$ cell (analogously, a $\ove C$
cell), then we are in the situations of Steps 5 or 6 as all other
possibilities are excluded by Step 7 and Table \ref{angletable}.   In case the ground state contains just
$Z$ or $\ove Z$, two cells of the same type have to share a bond and this has to be of
type $(\mu,\nu)$  (recall the orientations from Figure \ref{CCZZfigure}). The corresponding incidence angle is $0$ and,  after possible reorientation of $H$,  we are
in the situation of \eqref{step} (Step 5).
\end{proof}

\section*{Acknowledgement}
The support by the Austrian Science Fund (FWF) projects F\,65, P\,27052, and I\,2375 and the Alexander von Humboldt Foundation is gratefully acknowledged. This work has been funded by the Vienna Science and Technology Fund (WWTF)
through Project MA14-009. The authors acknowledge the kind hospitality of the Mathematisches Forschungsinstitut Oberwolfach, where part of this research was performed.

\bibliographystyle{alpha}

\end{document}